\def\year{2022}\relax
\documentclass[letterpaper]{article} 
\usepackage{aaai22}  
\usepackage{times}  
\usepackage{helvet} 
\usepackage{courier}  
\usepackage[hyphens]{url}  
\usepackage{graphicx} 
\usepackage{natbib}  
\usepackage{caption} 
\title{Toward Physically Realizable Quantum Neural Networks}

\author {
  Mohsen Heidari,   
  Ananth Grama, 
  Wojciech Szpankowski  \\
}
\affiliations {
   Department of Computer Science, Purdue University, West Lafayette, IN, USA
 \\
    \{mheidari, ayg, szpan\}@purdue.edu
}

\usepackage{amsmath}
\usepackage{amsfonts}
\usepackage{amsthm}   
\usepackage{mathrsfs}
\usepackage{tikz}
\usepackage[utf8]{inputenc}

\usepackage{epstopdf}
\usepackage{mathtools}
\usepackage{graphicx}
\usepackage{bbm}
\usepackage{enumerate}
\usepackage{epsf,verbatim,amssymb,array,cite,multicol,multirow}  
\usepackage{psfrag,bm,xspace}
\usepackage{hhline}
\usepackage{xstring}
\usepackage{ifthen}

\usepackage{xparse}
\usepackage{physics}

\ifdefined \theorem 
\else
  \newtheorem{theorem}{Theorem}
\fi
\newtheorem{lem}{Lemma}
\ifdefined \corollary 
\else
\newtheorem{corollary}{Corollary}
\fi

\ifdefined \proposition 
\else
\newtheorem{proposition}{Proposition}
\fi

\ifdefined \definition 
\else
\newtheorem{definition}{Definition}
\fi

\ifdefined \example 
\else
\newtheorem{example}{Example}
\fi

\ifdefined \remark 
\else
\newtheorem{remark}{Remark}
\fi

\makeatletter
\def\old@comma{,}
\catcode`\,=13
\def,{%
  \ifmmode%
    \old@comma\discretionary{}{}{}%
  \else%
    \old@comma%
  \fi%
}
\makeatother

\usepackage{xcolor}
\definecolor{darkblue}{rgb}{0.1,0.1,0.8}
\definecolor{DarkGreen}{rgb}{0,0.6,0}
\definecolor{brickred}{rgb}{0.8, 0.25, 0.33}
\definecolor{britishracinggreen}{rgb}{0.0, 0.26, 0.15}
\definecolor{calpolypomonagreen}{rgb}{0.12, 0.3, 0.17}
\definecolor{ao(english)}{rgb}{0.0, 0.5, 0.0}
	\definecolor{cadmiumgreen}{rgb}{0.0, 0.42, 0.24}
\definecolor{burgundy}{rgb}{0.5, 0.0, 0.13}
\usepackage{etoolbox}

\providetoggle{Blue_revision}
\settoggle{Blue_revision}{true}

\providetoggle{Track}
\settoggle{Track}{true}
\newcommand{\addv}[3]{%
	\iftoggle{Track}{%
    	\IfEqCase{#1}{%
       	 	{a}{\ifthenelse{\equal{#2}{ON}}{{\color{cadmiumgreen}#3}}{#3}}%
        	{b}{\ifthenelse{\equal{#2}{ON}}{{\color{brickred}#3}}{#3}}%
       		{c}{\ifthenelse{\equal{#2}{ON}}{{\color{burgundy}#3}}{#3}}%
       		{d}{\ifthenelse{\equal{#2}{ON}}{{\color{red}#3}}{#3}}
    	}[\PackageError{tree}{Undefined option to tree: #1}{}]%
	}{#3}%
}

 \usepackage{hyperref}

\hypersetup{
colorlinks=true, %
  pdfstartview={FitH},
    linkcolor=red,
    citecolor=blue,
    urlcolor={blue!80!black}
}

\usepackage{graphicx}

\newcounter{relctr} 
\everydisplay\expandafter{\the\everydisplay\setcounter{relctr}{0}} 

\AtBeginDocument{} 

\global\long\def\RR{\mathbb{R}}
\global\long\def\CC{\mathbb{C}}

\global\long\def\EE{\mathbb{E}}
\global\long\def\PP{\mathbb{P}}

\global\long\def\11{\mathbbm{1}}

\newcommand{\bfs}{\mathbf{s}}

\global\long\def\+{\oplus}

\newcommand\pmm{\{-1,1\}}

\def\<{\langle}
\def\>{\rangle}

\ifdefined \var 
  \renewcommand{\var}{\mathsf{var}}
\else
  \newcommand{\var}{\mathsf{var}}
\fi

\ifdefined \abs \else
 \newcommand{\abs}[1]{\lvert#1\rvert}
\fi

\ifdefined \norm \else
 \newcommand{\norm}[1]{\lVert#1\rVert}
\fi

\ifdefined \set 
  \renewcommand{\set}[1]{\left\{#1\right\}}
\else
  \newcommand{\set}[1]{\left\{#1\right\}}
\fi

\DeclareMathOperator*{\argmin}{arg\,min}

\usepackage{stackengine}

\DeclareMathOperator*{\tensor}{\otimes}

\providecommand{\tr}{tr}

\ifdefined \Tr 
  \renewcommand{\Tr}[1]{\tr \Big\{#1\Big\}}
\else
  \newcommand{\Tr}[1]{\tr \Big\{#1\Big\}}
\fi

\def\Loss{L}

\def\qas{a_{\bfs}}
\def\qps{\sigma^\bfs}

\def\L2{\mathcal{L}^2(\mathcal{X}^d, P_{X^d})}

\def\E_mu{\mathcal{E}_{\mu}(\epsilon')}

\usepackage[nolist,nohyperlinks]{acronym}

\newacro{ptp}[PtP]{Point-to-Point}
\newacro{iid}[i.i.d.]{independent and identically distributed} 
\newacro{IID}[i.i.d.]{independent and identically distributed} 
\newacro{UFFS}[UFFS]{Unsupervised Fourier Feature Selection}
\newacro{SFFS}[SFFS]{Supervised Fourier Feature Selection}
\newacro{LS}[LS]{Laplacian Score}
\newacro{MAE}[MAE]{mean absolute error}
\newacro{MSE}[MSE]{mean square error}
\newacro{PAC}[PAC]{\textit{probably approximately correct}}
\newacro{VC}[VC]{Vapnik–Chervonenkis}
\newacro{ERM}[ERM]{Empirical Risk Minimization}
\newacro{SVM}[SVM]{support-vector machine}
\newacro{POVM}[POVM]{positive operator-valued measure}
\newacro{QLD}[QLD]{quantum low-degree algorithm}
\newacro{QNN}[QNN]{quantum neural networks}
\newacro{SGD}[SGD]{stochastic gradient descent}
\usepackage{booktabs} 
\def\-{\raisebox{.75pt}{-}}
\usepackage[linesnumbered,ruled,vlined,procnumbered]{algorithm2e}

\newtheorem{innercustomgeneric}{\customgenericname}
\providecommand{\customgenericname}{}
\newcommand{\newcustomtheorem}[2]{%
  \newenvironment{#1}[1]
  {%
   \renewcommand\customgenericname{#2}%
   \renewcommand\theinnercustomgeneric{##1}%
   \innercustomgeneric
  }
  {\endinnercustomgeneric}
}

\newcustomtheorem{customthm}{Theorem}
\newcustomtheorem{customlemma}{Lemma}

\settoggle{Track}{false}
\settoggle{Blue_revision}{false}

\newcommand{\addmh}[1]{\addv{d}{ON}{#1}}
\begin{document}

\maketitle

\begin{abstract}
There has been significant recent interest in quantum neural networks (QNNs), along with their applications in diverse domains. Current solutions for QNNs pose significant challenges concerning their scalability, ensuring that the postulates of quantum mechanics are satisfied, and that the networks are physically realizable. The exponential state space of QNNs poses challenges for the scalability of training procedures. The no-cloning principle prohibits making multiple copies of training samples, and the measurement postulates lead to non-deterministic loss functions. Consequently, the physical realizability and efficiency of existing approaches that rely on repeated measurement of several copies of each sample for training QNNs is unclear. This paper presents a new model for QNNs that relies on band-limited Fourier expansions of transfer functions of quantum perceptrons (QPs) to design scalable training procedures. This training procedure is augmented with a randomized quantum stochastic gradient descent technique that eliminates the need for sample replication. We show that this training procedure converges to the true minima in expectation, even in the presence of non-determinism due to quantum measurement. Our solution has a number of important benefits: (i) using QPs with concentrated Fourier power spectrum, we show that the training procedure for QNNs can be made scalable; (ii) it eliminates the need for resampling, thus staying consistent with the no-cloning rule; and (iii) enhanced data efficiency for the overall training process since each data sample is processed once per epoch. We present detailed theoretical foundation for our models and methods' scalability, accuracy, and data efficiency. We also validate the utility of our approach through a series of numerical experiments. 

\end{abstract}

\section{Introduction}
%
%
%

With rapid advances in quantum computing, there has been increasing interest in quantum machine learning models and circuits. Work at the intersection of deep learning and quantum computing has resulted in the development of quantum analogs of classical neural networks, broadly known as quantum neural networks (QNNs) \cite{Mitarai2018,Farhi2018,Schuld_2020,Beer2020}. These networks are comprised of quantum perceptrons (QPs)~\cite{lewenstein1994quantum}, and the corresponding circuits (either layers of QPs or variational circuits) are physically realized through technologies such as trapped ions, quantum dots, and molecular magnets. QNNs have been applied to both classical and quantum data in the context of diverse applications. They have also inspired the development of novel classical neural network architectures~\cite{garg2020advances}, which have demonstrated excellent performance in applications such as image understanding and natural language processing.

Our effort focuses on developing QNNs for quantum data within the constraints of physically realizable quantum models (e.g., no-cloning, measurement state collapse). 
\addmh{Among many applications, quantum state classification
is of significance \cite{Chen2018} and has been studied under various specifications such as \textit{separability} of quantum states \citep{Gao2018,Ma2018}, integrated quantum photonics \citep{Kudyshev2020}, and dark matter detection through classification of polarization state of photons \cite{Dixit2021}. }
In yet other applications, it has been shown that data traditionally viewed in classical settings are better modeled in a quantum framework. In an intriguing set of results, Orus et al. demonstrate how data and processes from financial systems are best modeled in quantum frameworks~\cite{ORUS2019100028}.


The problem of constructing accurate QNNs for direct classification of quantum data is a complex one for a number of reasons: (i) the state of a QNN is exponential in the number of QPs. For this reason, even relatively small networks are hard to train in terms of their computation and memory requirements; (ii) the no cloning postulate of quantum states implies that the state of an unknown qubit cannot be replicated into another qubit; (iii) measurement of the state of a qubit corresponds to a realization of a random process, and results in state collapse; (iv) the non-deterministic nature of measurements implies that computation of classical loss functions itself is non-deterministic; and (v) existing solutions that rely on replicated data and oversampling for probabilistic bounds on convergence suffer greatly from lack of data efficiency. While there have been past efforts at modeling and training QNNs, one or more of these fundamental considerations are often overlooked. We present a comprehensive solution to designing quantum circuits for QNNs, along with provably efficient training techniques that satisfy aforementioned physical constraints.


\subsection{Main Contributions}
%
%

\addmh{In this work,} we propose a new QNN circuit, its analyses, and a training procedure to address key shortcomings of prior QNNs.
We make the following specific contributions: (i) we present a novel model for a QP whose Fourier spectrum is concentrated in a fixed sided band; (ii) we demonstrate how our QPs can be used to construct QNNs whose state and state updates can be performed in time linear in the number of QPs; (iii) we present a new quantum circuit that integrates gradient updates into the network, obviating the need for repeated quantum measurements and associated state collapse; (iv) we demonstrate data efficiency of our quantum circuit in that it does not need data replication to deal with stochastic loss functions associated with measurements; and (v) we show that our quantum circuit converges to the exact gradient in expectation, even in the presence of the stochastic loss function without data replication. We present detailed theoretical results, along with simulations demonstrating our model's scalability, accuracy, and data efficiency.


\subsection{Related Results}\label{sec:related}

Quantum neural networks have received significant recent research attention since the early
work of Toth et al.~\cite{Toth_1996}, who proposed an architecture comprised of quantum dots connected in
a cellular structure. 
Analogous to conventional neural networks, Lewenstein~\cite{lewenstein1994quantum}
proposed an early model for a quantum perceptron (QP) along with its corresponding unitary transformation. 
Networks of these QPs were used to construct early QNNs, which were used to classify suitably
preprocessed classical data inputs. \addmh{Since then, there has been significant research interest in the development of QNNs over the past two decades~\cite{Schuld_2014,Mitarai2018,Farhi2018,Torrontegui2018,Schuld_2020,Beer2020}. Massoli et al.~\cite{massoli2021leap} provide an excellent survey and classification of various models and technologies for QNNs.}

QNNs are generally trained using Variational Quantum
Algorithms~\cite{cerezo2020variational,guerreschi2017practical}, which use quantum analogs of
QNN parameters, initial qubits state (input data), and trail state (output of the QNN). These are
used in conjunction with optimization procedures such as gradient descent to update network parameters.
Among the best known of these procedures is the Quantum Approximate Optimization Algorithm of
Farhi et al.~\cite{farhi2014quantum}, which prescribes a circuit and an approximation bound for
efficient combinatorial optimization. This method was subsequently generalized to continuous
non-convex problems, for finding the ground state of a circuit through a suitably
defined evolution operator. A slew of recent results have used optimizers such as L-BFGS, Adam,
Nestorov, and SGD to train different QNN models. An excellent survey of these methods and associated
tradeoffs is presented by Massoli et al.~\cite{massoli2021leap}. The basic problem of exponential
network state and associated cost of optimization procedures is not directly addressed in these
results. In contrast, our focus in this work is on reducing the cost of training QNNs through
flexible and powerful constraints on Fourier expansions of QPs.

A second issue relating to convergence of these prior QNNs arises from the stochastic nature
of quantum measurements -- i.e., measuring the quantum state (in particular, the trial state) to compute
a loss function is inherently stochastic. For this reason, the convergence of prior QNNs relies
on repeated processing of the same data sample until the ground state is achieved with prescribed high
probability. 
There are two challenges
for these prior solutions: (i) repeated reads of the same state are problematic because of the no-cloning
rule for quantum states; and (ii) the cost of the training procedure is significantly higher due
to more extensive data volumes. In contrast, our method does not require data resampling, and therefore
does not have the aforementioned drawbacks.

Our work formally demonstrates how we can compute gradients in effective and efficient procedures (Theorem 1) and achieve asymptotically faster convergence in terms of data processing (Theorem 2). These two results provide the formal basis for our models and methods, as compared to prior results.


\section{Model}\label{sec:model}

\subsection{Quantum Machine Learning Model}
 The focus of this paper is on classification of quantum states. Following the framework in \cite{HeidariQuantum2021},  our quantum machine learning model inputs a set of $n$ quantum states corresponding to quantum training data samples, each paired with a classical label $y\in\mathcal{Y}$. We write these training samples as $\{(\rho_i, y_i)\}_{i=1}^n$. Each $\rho_i$ is a density operator acting on the Hilbert space $H$ of $d$-qubits, i.e., $\dim(H) = 2^d$. The data samples with the labels are drawn independently from a fixed but unknown probability distribution $D$. 
A quantum learning algorithm takes training samples as input to construct a predictor for labels of unseen samples. The prediction is in the form of a \textit{quantum measurement}  that acts on a quantum state and outputs $\hat{y}\in \mathcal{Y}$ as the predicted label. 

To test a predictor $\mathcal{M}:=\{M_{y}: y\in \mathcal{Y}\}$, a new sample $({\rho}_{{test}}, y_{test})$ is drawn randomly according to $D$. Without revealing $y_{test}$, we measure ${\rho}_{{test}}$ with $\mathcal{M}$. In view of the the postulate of quantum measurements, the outcome of this predictor is a random variable $\hat{Y}$ that together with the true label $Y$ form the joint probability distribution given by $$\PP\{Y=y_{test}, \hat{Y}=\hat{y}\}=D_{Y}(y_{test})\tr{M_{\hat{y}}\rho_{test}},$$
where $D_Y$ is the marginal of $D$.  
Thus, one can use an existing loss function $\ell: \mathcal{Y}\times \mathcal{Y}\mapsto \RR$  (such as the 0-1 loss or square loss) to measure the accuracy of the predictors.


\subsection{Quantum Neural Networks}
Figure \ref{fig:QNN model} shows a generic model for QNNs considered in our work.  The QNN consists of a \textit{feed-forward} network of interconnected QPs, each of which is a parametric unitary operator. At the output of the last layer is a quantum measurement acting on the readout qubits to produce a classical output as the label's prediction. 

The input to the QNN (or the first layer) is a state of $d'$ qubits consisting of the original $d$-qubit sample $\rho$ padded with additional auxiliary qubits, say $\ket{0}$. The padding allows the QNN to implement a more general class of operations (quantum channels) on the samples \footnote{ It is well-known that any quantum channel (CPTP map) is a partial trace of a unitary operator (Stinespring dilation) in a larger Hilbert space. }.  The input to the $l^{th}$ layer is the output of the previous layer $l-1$.  

\begin{figure}[hbtp]
\centering
\includegraphics[scale=0.7]{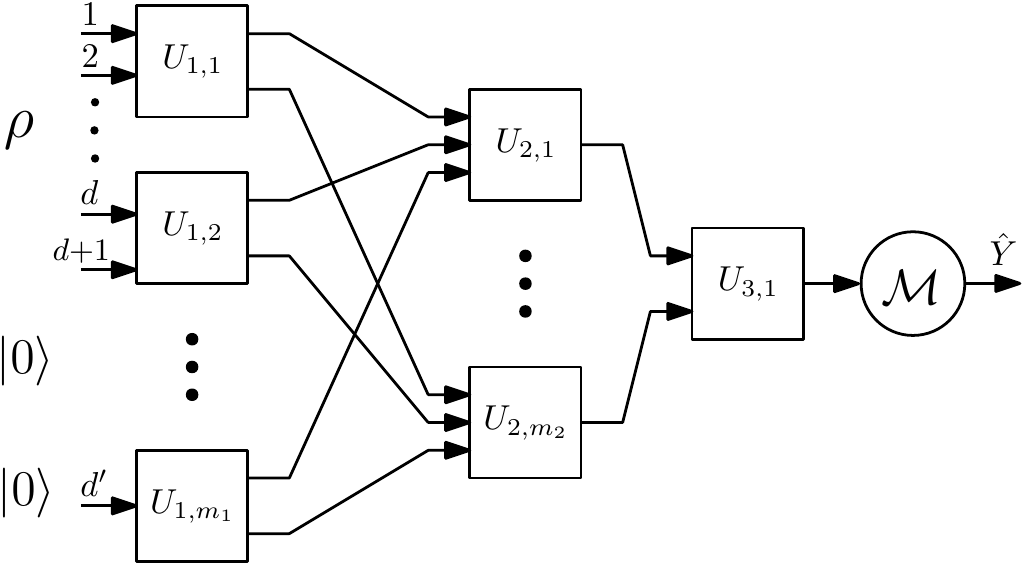}
\caption{A generic feed-forward QNN with one hidden layer and parametric unitary operators as QPs. 
}
\label{fig:QNN model}
\end{figure}


\subsection{Band Limited QPs}
We use an abstraction of a QP that generalizes the model of Farhi et al.~\cite{Farhi2018}. We consider QPs of the form $U=e^{i A}$, where $A$ is a Hermitian operator acting on a small subsystem of the $d'$-qubit system, hence it is band-limited. This notion stems from quantum Fourier expansion using Pauli operators. 

\noindent{\bf Quantum Fourier Expansion:} 
Tensor products of Pauli operators together with identity  have been used to develop a quantum Fourier expansion \cite{Montanaro2008} analogous to the Fourier expansion on the Boolean cube \cite{Wolf2008}. We denote these operators by  $\{\sigma^0,\sigma^1, \sigma^2, \sigma^3\}$. Furthermore, as a shorthand, given any vector $\bfs\in \set{0,1,2,3}^d$, we denote by $\qps$ the tensor product $\sigma^{s_1}\tensor \sigma^{s_2}\tensor \cdots \tensor \sigma^{s_d}.$ With this notation we present the notion of quantum Fourier expansion:
\begin{remark}\label{rem:quantum Fourier}
Any bounded  operator $A$ on the Hilbert space of $d$ qubits can be uniquely written as:   $A = \sum_{\bfs \in \{0,1,2,3\}^d} a_{\bfs} ~\sigma^{\bfs},$ where $a_\bfs\in \CC$ are the Fourier coefficients of $A$ {and are given by $a_\bfs=\frac{1}{2^d}\tr\big\{A\qps\big\}$.} If $A$ is Hermitian then $\qas \in \RR$. 
\end{remark}

Consider a QP acting on the subsystem corresponding to a subset of coordinates  $\mathcal{J}\subseteq [d']$. Since $A$ acts on qubits, then from Remark \ref{rem:quantum Fourier}, it decomposes in terms of Pauli operators as 
\begin{align*}
A = \sum_{\substack{\bfs: s_j = 0, \forall j \notin \mathcal{J}}} \qas~ \qps. 
\end{align*}
Note that the Fourier coefficients are zero outside of the coordinate subset $\mathcal{J}$, leading to a band-limited power spectrum in the Fourier domain. Moreover, these Fourier coefficients are used to fine-tune the QP. Intuitively, these coefficients can be viewed as quantum analogs of weights in classical NNs. 

To control the power spectrum of QPs, we bound $|\mathcal{J}|\leq k$, where $k$ is a bandwidth parameter. This constraint ensures that  $\qas=0$ for any $\bfs$ with more than $k$ non-zero components. In addition, it provides the basis for controlling the cost of gradient computation and update by limiting the number of coefficients to optimize over. In this paper, we typically set $k=2$.


Let $\overrightarrow{a}_{l, j}$ denote the vector of Fourier coefficients for the $j$th QP in the $l^{th}$ layer and $U_{l,j}(\overrightarrow{a}_{l, j})$ denote the  unitary operator of this QP. With this setup, each layer consists of several QPs acting on different subsystems of $\dim \leq 2^k$. By $U_l(\overrightarrow{a}_{l})$, denote the overall unitary operator of the $l^{th}$ layer, which is given by 
$
U_l(\overrightarrow{a}_{l}) = \prod_{j=1}^{m_l} U_{l,j}(\overrightarrow{a}_{l, j}),
$
   where $m_l$ is the number of QPs in this layer.
 Building on this formulation, a QNN with $L$ layers is itself a unitary operator that factors as:
\begin{align*}
  U_{QNN}(\overrightarrow{a}) = U_L(\overrightarrow{a}_{L})U_{L-1}(\overrightarrow{a}_{L-1})\cdots U_1(\overrightarrow{a}_{1}),
  \end{align*}
  where $\overrightarrow{a}$ is the vector of all the parameters of all the layers. Since, each layer is characterized by $m_l 4^k$ coefficients,  the QNN needs at most $c_{QNN} = 4^k m $ coefficients to be optimized over, where $m$ is the total number of QPs. 

\subsection{Computing the Loss}
To the output of the final layer, one applies a quantum measurement $\mathcal{M}=\{M_y: y\in \mathcal{Y}\}$ to produce the prediction for the classical label. For a padded sample $\rho' = \rho \tensor \ketbra{0}$, measuring the output qubit of the QNN results in an outcome $\hat{y}$ with the following probability:
\begin{align*}
\PP(\hat{Y}=\hat{y}| \overrightarrow{a}, \rho) = \tr{ M_{\hat{y}} U_{QNN}(\overrightarrow{a}) \rho' U^\dagger_{QNN}(\overrightarrow{a})}.
\end{align*}
 
Let   $\ell: \mathcal{Y}\times \mathcal{Y}\mapsto \RR$ be the loss function. Then, conditioned on a fixed sample  $(\rho, y)$,  the sample's expected loss taken w.r.t. $\hat{Y}$ is given by
\begin{align}\nonumber
L( &\overrightarrow{a}, \rho, y)  = \EE_{\hat{Y}} [ \ell(\hat{Y}, y)]\\
& = \sum_{\hat{y}}\ell(y, \hat{y}) \tr{ M_{\hat{y}} U_{QNN}(\overrightarrow{a}) \rho' U^\dagger_{QNN}(\overrightarrow{a})}. 
\end{align}
Taking the expectation over the sample's distribution $D$ gives the generalization loss:
\begin{align*}
L(\overrightarrow{a}) &=  \EE_D[L(\overrightarrow{a}, \rho, Y)].
\end{align*}

While it is desirable to minimize the expected loss, this goal is not feasible, since the underlying distribution $D$ is unknown. Alternatively,  given a sample set $\mathcal{S}_n = \{(\rho_j, y_j)\}_{j=1}^n$, one aims to minimize the average per-sample expected loss:
\begin{align}\label{eq: average sample exp loss }
\frac{1}{n}\sum_j L( &\overrightarrow{a}, \rho_j, y_j).
\end{align}
However, unlike classical supervised learning, exact computation of this loss is not possible as $\rho_j$'s are unknown in general. One can approximate this loss, but under the condition that several exact copies of each sample are available. Typically, this is infeasible in view of the no-cloning principle. This restriction makes training more challenging than its classical counterpart. We argue that even if several copies of the samples are available, such a strategy is not desirable due to its increased data and computational cost. In the next section, we present our optimization approach based on a randomized variant of SGD, to address this problem.

\section{Quantum Stochastic Gradient Descent}

In this section, we present our approach to training QNNs using gradient descent. Ideally, if the $t^{th}$ sample's expected loss $\Loss(\overrightarrow{a},\rho_t, y_t)$ was known, one would apply a standard gradient descent method to train the QNN as follows:
\begin{align*}
\overrightarrow{a}^{(t+1)} = \overrightarrow{a}^{(t)} - \eta_t \nabla \Loss(\overrightarrow{a},\rho_t, y_t).
\end{align*}
Typically, $\Loss(\overrightarrow{a},\rho_t, y_t)$ is unknown hence the above update is not practical. One approach to deal with this issue is to use several exact copies of each sample, enabling one to approximate the gradient ~\cite{Farhi2018}. 
The drawback of this approach is two-fold: (i) multiple exact copies of each sample are needed; and (ii) the associated utilization complexity for training is high. The latter drawback derives from the fact that approximating the derivative of the sample's expected loss with error upto $\varepsilon$ requires processing of $O(1/\varepsilon^2)$ exact copies. Therefore, training a QNN with the total of $c_{QNN}$ parameters and using $T$ iterations of SGD requires $O( T\frac{c_{QNN}}{\varepsilon^2})$ uses of the QNN. 

 In this paper, we propose an alternate procedure for performing quantum stochastic gradient descent (QSGD) without the need for exact copies, substantially reducing the utilization complexity to $O(Tc_{QNN})$. We note here that minimizing the loss without access to exact copies is a more complex problem since even the per-sample loss is unknown due to the randomness from quantum measurements. 
 In particular, we introduce a randomized QSGD with the following update rule:
 \begin{align*}
\overrightarrow{a}^{(t+1)} = \overrightarrow{a}^{(t)} - \eta_t Z_t,
\end{align*}
where $Z_t$ is a random variable representing the outcome of a gradient measurement that is designed in the preceding section of the paper.  
We present an analysis for a QNN and associated convex objective function. We note that non-convex objective functions associated with deep networks are typically formulated as convex quadratic approximations within a trust region, or using cubic regularization.

\subsection{Unbiased Measurement of the Gradient}

 We now present a procedure for measuring the derivative of the expected loss (Figure \ref{fig:QNN gradient}). 
 We start by training a single-layer QNN. Training multi-layer QNNs with this approach is done using \textit{backpropagation} to find (local) minima for the loss function. 
 The corresponding unitary operator of a single-layer network is of the form: 
\begin{align}\label{eq:perceptron U1}
U_{QNN}(\overrightarrow{a}) = U_1(\overrightarrow{a}) = \exp\Big\{i\sum_{\bfs} \qas \qps \Big\}.
\end{align}
In what follows, we intend to measure the derivative of the loss w.r.t. $\qas$ for a given $\bfs$ appearing in the summation above. In our method, at each time $t$, we apply the QNN on the $t^{th}$ sample $\rho_t$, which produces $\rho_t^{out}$. Next, we add an auxiliary qubit to $\rho^{out}_t$, which creates the following state:
 \begin{align}\label{eq:rho tild}
 \tilde{\rho}_t = \rho^{out}_t \tensor \ketbra{+}_E,
 \end{align}
 where $\ket{+}_E=\frac{1}{\sqrt{2}}(\ket{0}_E+\ket{{1}}_E)$, and $E$ represents the auxiliary Hilbert space. Then, we apply the following unitary operator on $\tilde{\rho}_t$:
 \begin{align}\label{eq:Vs}
 V_\bfs =  e^{i\frac{\pi}{4} \qps} \tensor \ketbra{0}_E + e^{-i\frac{\pi}{4} \qps} \tensor \ketbra{1}_E.
 \end{align}

\begin{figure}[tp]
\centering
\includegraphics[scale=1]{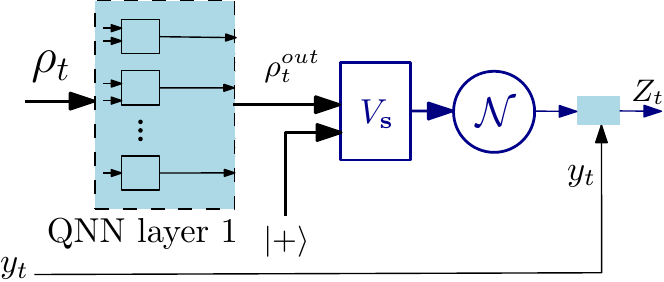}
\caption{The procedure for measuring the derivative of the loss consisting of a unitary operator $V_\bfs$ followed by a quantum measurement $\mathcal{N}$ with additional classical processing. 
}
\label{fig:QNN gradient}
\end{figure}

Next, we measure the state by a quantum measurement $\mathcal{N}:=\set{\Lambda_{b,\hat{y}}: \hat{y}\in \mathcal{Y}, b\in \{0,1\}}$ with outcomes in $\mathcal{Y}\times \{0,1\}$ and operators:
\begin{align}\label{eq:POVM N}
\Lambda_{b, \hat{y}}  = M_{\hat{y}} \tensor \ketbra{b}_E, \qquad \forall \hat{y}\in \mathcal{Y},~ b =0,1. 
\end{align}
Finally, if the outcome of this measurement is $(\hat{y},b)$, we compute $z_t =-2(-1)^b \ell(y,\hat{y})$ as the measured gradient, where $\ell(\cdot,\cdot)$ is the loss function. Note that $Z_t$ is a random variable depending on the current parameters of the network and the input sample  $(\rho_t, y_t)$. The following lemma,  shows that $Z_t$ is an unbiased measure of the derivative of loss with respect to $\qas$.  

\begin{lem}\label{lem:gradient estimation}
Let $Z_t$ be the output of the procedure shown in Figure \ref{fig:QNN gradient} applied on the $t^{th}$ sample. Then $Z_t$ is an unbiased estimate of the derivative of the loss. More precisely,  
$\EE[Z_t|y_t, \rho_t] = \frac{\partial \Loss( \overrightarrow{a}, \rho_t, y_t)}{\partial \qas}.$
\end{lem}
\begin{proof}
We start by taking the derivative of the sample's expected loss w.r.t. $\qas$. Instead of the finite difference method, we compute the derivative directly. Using \eqref{eq:perceptron U1},  for the $t^{th}$ sample, we can write:
\begin{align*}
&\frac{\partial  \Loss( \overrightarrow{a}, \rho_t, y_t)}{\partial \qas} =\sum_{\hat{y}} i\ell(y_t,\hat{y}) \tr\Big\{ M_{\hat{y}}\Big(\\
&\qps U_{QNN}(\overrightarrow{a})\rho_t U^t_{QNN}(\overrightarrow{a})
 - U_{QNN}(\overrightarrow{a})\rho_t U^t_{QNN}(\overrightarrow{a})\qps\Big) \Big\}.
\end{align*}
Denote $\rho_t^{out}=U_{QNN}(\overrightarrow{a})\rho_t U^\dagger_{QNN}(\overrightarrow{a})$ as the output state. Hence, the derivative is given by:
\begin{align}\label{eq:Delta Loss}
\frac{\partial  \Loss( \overrightarrow{a}, \rho_t, y_t)}{\partial \qas} 
& = \sum_{\hat{y}} i \ell(y_t,\hat{y})  \tr{M_{\hat{y}} \big[ \qps, \rho_t^{out}\big]},
\end{align}
where $[\cdot, \cdot]$ is the commutator. Note that $\big[ \qps, \rho_t^{out}\big]$ might not be measurable directly. However, as $\qps$ is a product of Pauli operators, as shown in \cite{Mitarai2018}, we have that: 
\begin{align*}
\big[ \qps, \rho_t^{out}\big] &= i \Big( e^{-i\frac{\pi}{4}\qps}\rho_t^{out}e^{i\frac{\pi}{4}\qps} - e^{i\frac{\pi}{4}\qps}\rho_t^{out}e^{-i\frac{\pi}{4}\qps} \Big).
\end{align*}
With this relation, the derivative of the loss equals:
\begin{align*}
\frac{\partial  \Loss( \overrightarrow{a}, \rho_t, y_t)}{\partial \qas} = -\sum_{\hat{y}}  \ell(y_t,\hat{y})  &\tr\Big\{ M_{\hat{y}} \Big( e^{-i\frac{\pi}{4}\qps}\rho_t^{out}e^{i\frac{\pi}{4}\qps}\\
& - e^{i\frac{\pi}{4}\qps}\rho_t^{out}e^{-i\frac{\pi}{4}\qps} \Big)\Big\}.
\end{align*}
Next, we show that the procedure in Figure \ref{fig:QNN gradient} measures this derivative. With the definition of $V_\bfs$ in \eqref{eq:Vs} and $\tilde{\rho}_t$ in \eqref{eq:rho tild}, we have that 
\begin{align*}
V_\bfs \tilde{\rho}_t V^\dagger_\bfs = \frac{1}{2}\Big( & e^{i\frac{\pi}{4} \qps}  \rho^{out}_t  e^{-i\frac{\pi}{4} \qps} \tensor \ketbra{0}\\
&+ e^{i\frac{\pi}{4} \qps}  \rho^{out}_t  e^{i\frac{\pi}{4} \qps} \tensor \outerproduct{0}{1}\\
& + e^{-i\frac{\pi}{4} \qps}  \rho^{out}_t  e^{-i\frac{\pi}{4} \qps} \tensor \outerproduct{1}{0}\\
&+ e^{-i\frac{\pi}{4} \qps}  \rho^{out}_t  e^{i\frac{\pi}{4} \qps} \tensor \outerproduct{1}{1}\Big).
\end{align*}
Next, we apply the measurement $\mathcal{N}$. Then, with $(\hat{y},b)$ as the outcome, we output $z_t =-2 (-1)^b \ell(y_t, \hat{y})$. It is not difficult to verify that conditioned on a fixed $\hat{y}$, the expected outcome of $(-1)^b$ equals:
\begin{align*}
\tr{(\Lambda_{0,\hat{y}}-\Lambda_{1,\hat{y}}) V_\bfs \tilde{\rho}_t V_\bfs^\dagger}& = \frac{1}{2} \tr{ M_{\hat{y}}~ e^{i\frac{\pi}{4} \qps}  \rho^{out}_t  e^{-i\frac{\pi}{4} \qps}}\\
&-\frac{1}{2}\tr {M_{\hat{y}}~ e^{-i\frac{\pi}{4} \qps}  \rho^{out}_t  e^{i\frac{\pi}{4} \qps}}.
\end{align*}
Therefore, the expectation of $Z_t$ is obtained by taking the expectation over $\hat{Y}$and equals to 
\begin{align*}
\EE[Z_t|y_t, \rho_t] & = -2\sum_{\hat{y}, b} (-1)^b \ell(y_t,\hat{y}) \tr{\Lambda_{b,\hat{y}} V_\bfs \tilde{\rho}_t V^\dagger_\bfs} \\
&=-\sum_{\hat{y}} \ell(y_t,\hat{y}) \tr\Big\{ M_{\hat{y}} \Big( e^{i\frac{\pi}{4} \qps}  \rho^{out}_t  e^{-i\frac{\pi}{4} \qps}\\
&\hspace{80pt} - e^{-i\frac{\pi}{4} \qps}  \rho^{out}_t  e^{i\frac{\pi}{4} \qps}\Big)\Big\} \\
&=\frac{\partial  \Loss( \overrightarrow{a}, \rho_t, y_t)}{\partial \qas}. 
\end{align*}
With that the proof is completed.
\end{proof}

\paragraph{Measuring the gradient with randomization:}
We have, thus far, presented a method to measure the derivative of the loss without sample reuse. Next, we present a randomized approach to measure the gradient with respect to all parameters $\overrightarrow{a}$. At each step, we randomly select a $\bfs$ as a component of $\overrightarrow{a}$ and measure the derivative of the loss with respect to it. Conditioned on a selected $\bfs$, we create a vector $\overrightarrow{Z}_t$ whose components are zeros except at the $\bfs^{th}$ component, which is the outcome of  measuring the derivative with respect to $\qas$. Consequently, by randomly selecting $\bfs$ and measuring the derivative with respect to it, we obtain an unbiased measure of the gradient. This procedure is summarized in Algorithm \ref{alg:QSGD}. With this argument and Lemma \ref{lem:gradient estimation}, we get the desired result which is formally stated below.

\begin{algorithm}[ht]
\caption{Randomized QSGD}\label{alg:QSGD}
\DontPrintSemicolon
\LinesNumbered
\KwIn{Training data $\{(\rho_t, y_t)\}_{t=1}^n$, Learning rate $\eta_t$}
\KwOut{Finial parameters of QNN: $\overrightarrow{a}$}

Initialize the parameter $\overrightarrow{a}$ with each component selected with uniform distribution over $[-1,1]$.\;
\For{$t=1$ to $n$}{
Randomly select a QP in the network. Let it be the $j^{th}$ QP of layer $l$.\;
Pass $\rho_t$ through the previous layers and let $\rho_t^{l-1}$ be the output of the layer $(l-1)$.\;
Randomly select a component $a_{l,j, s}$ of $\overrightarrow{a}_{l,j}$.\;
Apply $V_\bfs$ as in \eqref{eq:Vs}  on $\rho_t^{l-1}\tensor \ketbra{+}$ and measure the resulting state with $\mathcal{N}$ as in \eqref{eq:POVM N}. \;
With $(\hat{y},b)$ being the outcome, compute the measured derivative as $z_t =-2(-1)^b \ell(y,\hat{y})$.\;
Update the $s^{th}$ component of $\overrightarrow{a}_{l,j}$ as 
$a_{l,j, s} \leftarrow a_{l,j, s} - \eta_t z_t$.\;
}
\Return $\overrightarrow{a}$
\end{algorithm}

\begin{theorem}\label{thm:gradient}
Given a QNN with $c_{QNN}$ parameters (components) in $\overrightarrow{a}$. Randomized Quantum SGD  in Algorithm \ref{alg:QSGD} outputs a vector $\overrightarrow{Z}_t$, whose expectation conditioned on the $t^{th}$ sample satisfies: 
\begin{align*}
\EE[\overrightarrow{Z}_t| \rho_t, y_t]= \frac{1}{c_{QNN}}\nabla \Loss(\overrightarrow{a},\rho_t, y_t),
\end{align*}
where $c_{QNN}$ is the number of components in $\overrightarrow{a}$. 
\end{theorem}
\begin{proof}
Let us index all the Fourier coefficients of all the QPs of the network. For that denote by $a_{l,j,\bfs}$, the Fourier coefficient of the $j$th QP at layer $l$, where $l\in[1:L]$, $j\in [1:m_l]$ and $\bfs \in \set{0,1,2,3}^k$. With this notation, let $Z_{t,l,j,\bfs}$ denote the outcome of the gradient measure for measuring the derivative of the loss w.r.t $a_{l,j,\bfs},$ that is $\frac{\partial L(\overrightarrow{a},\rho_t, y_t)}{\partial a_{l,j,\bfs}}$. Furthermore, let $\overrightarrow{Z}_{t,l,j,\bfs}\in \RR^{c_{QNN}}$ denote a vector which is zero everywhere except at the coordinate corresponding to $(l,j,\bfs)$ which is equal to $Z_{t,l,j,\bfs}$. Similarly, let $\overrightarrow{\frac{\partial L(\overrightarrow{a},\rho_t, y_t)}{\partial a_{l,j,\bfs}}}$ be a vector in $\RR^{c_{QNN}}$ which is zero everywhere except at the coordinate corresponding to $(l,j,\bfs)$ which is equal to  $\frac{\partial L(\overrightarrow{a},\rho_t, y_t)}{\partial a_{l,j,\bfs}}$. 
Using this notation, we have that 
\begin{align*}
\EE[\overrightarrow{Z}_t |  \rho_t, y_t] & \stackrel{(a)}{=} \frac{1}{c_{QNN}}\sum_{l=1}^L\sum_{j=1}^{m_l}\sum_{\bfs \in \{0,1,2,3\}^k} \EE[ \overrightarrow{Z}_{t,l,j,\bfs} | \rho_t, y_t]\\
&\stackrel{(b)}{=}\frac{1}{c_{QNN}}\sum_{l=1}^L\sum_{j=1}^{m_l}\sum_{\bfs \in \{0,1,2,3\}^k} \overrightarrow{\frac{\partial L(\overrightarrow{a},\rho_t, y_t)}{\partial a_{l,j,\bfs}}}\\
&\stackrel{(c)}{=}\frac{1}{c_{QNN}} \nabla \Loss(\overrightarrow{a},\rho_t, y_t),
\end{align*}
where (a) holds because as described in Algorithm \ref{alg:QSGD} for generating $\overrightarrow{Z}_t$ each time an index $(l,j,\bfs)$ is selected at random out of the total of $c_{QNN}$ indeces and updated accordingly. Equality (b) is due to Lemma \ref{lem:gradient estimation}. Lastly, equality (c) follows from the definition of the gradient and the notation used for $\overrightarrow{\frac{\partial L(\overrightarrow{a},\rho_t, y_t)}{\partial a_{l,j,\bfs}}}$. With that the proof is complete.
\end{proof}

With this result, we obtain our SGD update rule as 
 \begin{align*}
\overrightarrow{a}^{(t+1)} = \overrightarrow{a}^{(t)} - \eta_t \overrightarrow{z_t},
\end{align*}
where $\eta_t$ is the learning rate containing the normalization effect of $c_{QNN}$.

\section{Discussion and Analysis}

\subsection{Convergence Rate}

We now present theoretical analysis of our proposed approach. Specifically, we show that randomized QSGD converges to the optimal choice of parameters when the underlying objective function is bounded and convex. Moreover, we argue that randomized QSGD is sample efficient as compared to the other QSGD methods, where the gradient is approximated using repeated measurements. Specifically, it has faster convergence rate as function of the number of samples/ copies.    We note here that, similar to the classical NNs, the convexity requirement is not generally satisfied in QNNs. However, our analysis with the convexity assumption provides insights on the convergence rate even for non-convex (DNN) objective functions solved using trust region, cubic regularization, or related techniques. Our numerical results in the next section demonstrate the convergence of randomized QSGD for training QNNs. 
 
 With Theorem \ref{thm:gradient}, and the standard techniques for analysis of the convergence rate of classical SGD  \cite{Shalev2014}, we obtain the following result:  
\begin{theorem}
Suppose the loss function $\ell(y, \hat{y})$ is bounded by $\gamma\in \RR$. Suppose also that for any choice of $(\rho, y)$, the sample's expected loss $\Loss(\overrightarrow{a},\rho, y)$ is a convex function of $\overrightarrow{a}$ with $\norm{\overrightarrow{a}}=1$. If our randomized SGD (Algorithm \ref{alg:QSGD}) is performed for $T$ iterations and $\eta = \frac{1}{2 \gamma\sqrt{T}}$, then the following inequality is satisfied
\begin{align*}
\EE[\Loss(\overrightarrow{a}_{ave},\rho, y)]- \Loss(\overrightarrow{a}^*) \leq \frac{2\gamma c_{QNN} }{\sqrt{T}},
\end{align*}
where $\overrightarrow{a}_{ave} = \frac{1}{T}\sum_t \overrightarrow{a}^{(t)}$,  $c_{QNN}$ is the number of parameters,  and $\overrightarrow{a}^* = \argmin_{\overrightarrow{a}: \|{\overrightarrow{a}}\| = 1}\Loss(\overrightarrow{a})$. 
\end{theorem} 
\begin{proof}
The theorem follows from Theorem \ref{thm:gradient} the following result.

\begin{customthm}{14.8 \cite{Shalev2014}}
Let $f$ be a convex function and $w^* = \argmin_{w:\norm{w}\leq B} f(w)$. Let $\{\overrightarrow{V}_t\}_{t=1}^T$ be a sequence of random vectors from SGD satisfying $\EE[\overrightarrow{V}_t | w^{(t)}] = \nabla f(w^{(t)})$, and $\norm{V_t}\leq \alpha$ with probability 1 for all $t$. Then, running SGD for $T$ iterations with $\eta = \frac{B}{\alpha \sqrt{T}}$ gives the following inequality 
\begin{align*}
\EE[f(w_{ave})] - f(w^*) \leq \frac{B\alpha}{\sqrt{T}},
\end{align*}
where $w_{ave} = \frac{1}{T}\sum_t w^{(t)}.$
\end{customthm}
As a result, from Theorem 1 of the main paper, setting $\overrightarrow{V}_t=c_{QNN}\overrightarrow{Z}_t$ gives a sequence of random vectors satisfying $\EE[\overrightarrow{V}_t| \rho_t, y_t]= \nabla \Loss(\overrightarrow{a},\rho_t, y_t)$. Since component of $\overrightarrow{Z}_t$ are zero except at one coordinate taking values from $-2(-1)^{b}\ell(y, \hat{y})$, then 
\begin{align*}
\norm{V_t} \leq c_{QNN} \norm{\overrightarrow{Z}_t} \leq 2 \gamma c_{QNN} ,
\end{align*}
where $\gamma = \max_{y,\hat{y}} \abs{\ell(y, \hat{y})}.$ Consequently, the proof is complete by applying the above theorem with $B=1$ and $\alpha = 2 \gamma c_{QNN}$. 
\end{proof}

 As a result, the convergence rate of randomized QSGD scales with the number of parameters $c_{QNN}$. Next, we compare this result to the previous methods in which the gradient is approximated by repeated measurement on several copies of the samples. First, we bound the number of copies needed to approximate the gradient of $\Loss(\overrightarrow{a},\rho, y)$ . 

\begin{lem}\label{lem:gradient approx}
Given $\varepsilon, \delta \in (0,1)$, with probability $(1-\delta)$, the gradient of a sample's expected loss is  approximated with error upto $\varepsilon$ and by measuring $O( \frac{1}{\varepsilon^2}c_{QNN}\log \frac{c_{QNN}}{\delta})$ copies.
\end{lem} 
 \begin{proof}[Proof Sketch]
 The proof follows from standard argument using Hoeffding's inequality with the difference that each copy can be used only once. Hence, each component of the gradient is estimated with a fraction $\frac{1}{c_{QNN}}$ of the total copies. As a result, the probability that  at least one component is estimated with error greater than $\varepsilon$ is bounded by $c_{QNN}\exp\{-\varepsilon^2 \frac{n}{c_{QNN}}\alpha\}$, where $\alpha$ is a constant depending on the loss function. We want this probability to be less than $\delta$. Thus, equating the bound with $\delta$ gives the required number of copies. 
 \end{proof}
Consequently, for a fair comparison with our approach, we fix the total number of samples/ copies. With $T$ samples randomized QSGD runs with $T$ iterations; whereas the prior repeated sample method with gradient approximation runs with $O(\frac{T\varepsilon^2}{c_{QNN}\log c_{QNN}})$ iterations. Therefore, given $T$ samples/ copies and under the convexity assumption, QSGD with gradient approximation has the following convergence rate:
 \begin{align*}
\EE[\Loss(\overrightarrow{a}_{ave},\rho, y)] - \Loss(\overrightarrow{a}^*) \leq \order{\frac{\sqrt{c_{QNN}\log c_{QNN}}}{\epsilon \sqrt{T}}}.
\end{align*}
Hence, for QNNs with moderate number of parameters (e.g., near term QNNs), randomized SGD converges faster even when sample duplication is not an issue. Faster convergence holds as long as $\frac{\log c_{QNN}}{c_{QNN}} = O(\epsilon^2)$. For instance, with $\varepsilon \approx 0.05$, the number of parameters can be $c_{QNN}\approx 3000$.
For larger networks, one might consider a hybrid approach.

\subsection{The Expressive Power of QNNs}
It is well known that classical neural networks are universal function approximators. For quantum analogs of this result, we can show that QNNs implement every quantum measurement on qubits. 
\begin{theorem}
For any quantum measurement $\mathcal{M}$ on the space of $d$ qubits, there is a QNN on the space of $d'>d$ that implements $\mathcal{M}$. Moreover, it is sufficient to use QPs with narrowness of $k=2$.      
\end{theorem}
\begin{proof}[Proof Sketch]
It is known that any quantum measurement with finite number of outcomes can be written as a quantum channel followed by a measurement in the standard basis on the readout qubits \cite{Wilde2013}. Hence, from Stinespring's dilation theorem \cite{Holevo2012}, such quantum channel can be written as $\tr_E\{V (\rho\tensor \ketbra{e}_E) V^\dagger\}$, where $V$ is a unitary operator on the padded space of $d'$ qubits. In addition, such unitary operator can be implemented using one perceptron with $k=d'$. The second statement follows from the fact that quantum 2-gates are universal. 
\end{proof}

\section{Numerical Results}\label{sec:experiments}

We experimentally assess the performance of our QNNs trained using randomized QSGD in Algorithm \ref{alg:QSGD}. In our experiments, we focus on binary classification of quantum states with $\set{-1,1}$ as the label set and with conventional 0-1 loss to measure predictor's accuracy. 

\subsection{Dataset}

We use a synthetic dataset from recent efforts \cite{Mohseni2004,Chen2018,Patterson2021,Li_Song_Wang_2021} focused on quantum state discrimination. In this dataset, for each pair of parameters $u,v\in [0,1]$, three input states are defined as follows:
\begin{align*}
 \ket{\phi_u} &= \sqrt{1-u^2}\ket{00}+u\ket{10},\\
\ket{\phi_{\pm v}} &= \pm \sqrt{1-v^2}\ket{01}+v\ket{10}.
 \end{align*} 
Using these input states, the two quantum states to be classified are:  
\begin{align*}
\rho_1(u) &=  \ketbra{\phi_u},\\
\rho_2(v) &= \frac{1}{2} \Big( \ketbra{\phi_{+v}} + \ketbra{\phi_{- v}}\Big).
\end{align*}
 We assign label $y=\-1$ for the first state and $y=1$ for the second state. For each sample, we generate a state $\rho_1$  with probability $p = 1/3$ and state $\rho_2$ with probability $(1-p)=2/3$. Furthermore, for each sample, parameters $u$ and $v$ are selected randomly, independently, and with the uniform distribution on $[0,1]^2$. Hence, there are infinitely many realizations of samples, and no sample replication is allowed.
 
 \subsection{QNN Setup}
 
 The QNN we use is shown in Figure \ref{fig:QNN exp}. It has two layers with two quantum QPs in the first layer and one quantum QP in the second layer. The input states are padded with two auxiliary qubits $\ket{0}\tensor\ket{0}$; hence, $d=2$ and $d'=4$. The bandwidth parameter is set to $k=2$ for each QP. The measurement used in the readout qubits has two outcomes in $\pmm$, each measured along the computational basis as  $$\mathcal{M}=\Big\{\ketbra{00}+\ketbra{11},~~ \ketbra{01}+\ketbra{10}\Big\}.$$

\begin{figure}[hbtp]
\centering
\includegraphics[scale=0.9]{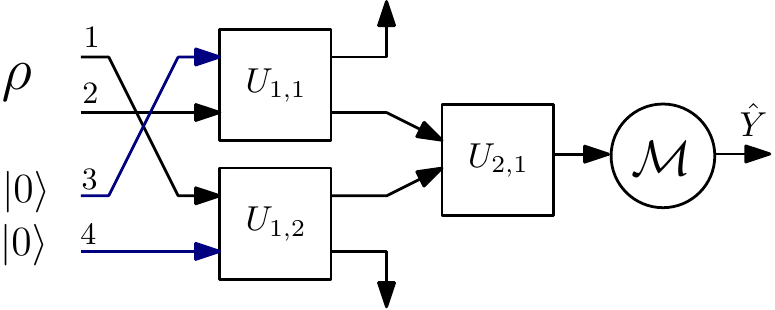}
\caption{The QNN used in the experimental study. }
\label{fig:QNN exp}
\end{figure}
\subsection{Experimental Results}

We train the QNN in Figure \ref{fig:QNN exp} using Algorithm \ref{alg:QSGD} and with $\eta_t=\frac{\alpha}{\sqrt{t}} $ and $\alpha \approx 0.77$.
Our training process has no epochs in contrast to conventional SGD, since sample replication is prohibited. Therefore, to show progress during the training phase, we group samples into multiple batches, each of $100$ samples. For each batch, we compute the average empirical loss of the QNN with updated parameters. Figure \ref{fig:Numerical Results} shows the loss during the training phase averaged over each batch. In addition to the empirical loss, we compute the average of the sample's expected loss $L( \overrightarrow{a}^{(t)}, \rho_t, y_t)$ for all samples in the batch (as in \eqref{eq: average sample exp loss }).  Furthermore, we compare the performance of the network with the optimal expected loss within each batch. The following result provides a closed-form expression for the optimal loss:

\begin{lem}\label{lem:optimal loss}
Given a set of $m$ labeled density operators $\{(\rho_j, y_j)\}_{j=1}^m$ with $y_j \in \pmm$, the minimum of the expected sample loss averaged over the $m$ samples as in \eqref{eq: average sample exp loss } equals:
$\frac{1}{2}\Big(1-\Big\| \frac{1}{m} \sum_j y_j \rho_j\Big\|_1\Big),$ 
where $\norm{\cdot}_1$ is the trace norm and the minimum is taken over all measurements applied on each $\rho_j$ for predicting $y_j$. 
\end{lem}
\begin{proof}
The argument follows from the Holevo–Helstrom theorem \cite{Holevo2012}. Note that we can assign two type of density operators for each batch: one with label $y=1$ and the other with $y=-1$. These density operators are 
\begin{align*}
\rho_- = \frac{1}{m_{-}} \sum_{j} \rho_j \11\{y_j = -1\}\\
\rho_+ = \frac{1}{m_{+}} \sum_{j} \rho_j \11\{y_j = +1\},
\end{align*}
where $m_{-}$ and $m_{+}$ are the number of $\rho_j$'s with label $y_j=-1$ and $y=+1$, respectively. Therefore, one can assume that $\rho_-$ occurs with probability $q = \frac{m_-}{m}$ and $\rho_+$ occurs with probability $(1-q) = \frac{m_+}{m}$. With this setup, the expected loss averaged over the $m$ samples equals to  probability of error for distinguishing between $\rho_-$ and $\rho_+$. 
From the Holevo–Helstrom theorem \cite{Holevo2012}, the maximum probability of success for distinguishing between the two types of density operators is 
\begin{align*}
P^{max}_{success} = \frac{1}{2}\big(1+\norm{q \rho_- - (1-q)\rho_+}_1\big).
\end{align*}
It is not difficult to check that the terms inside the trace norm equal to $\frac{1}{m} \sum_j y_j \rho_j$ that is the desired expression as in the statement of the lemma. Thus, the minimum probability of error (expected loss) is $(1-P^{max}_{success})$ as in the statement of the lemma.
\end{proof}

\begin{figure}[hbtp]
\centering
\includegraphics[scale=0.3]{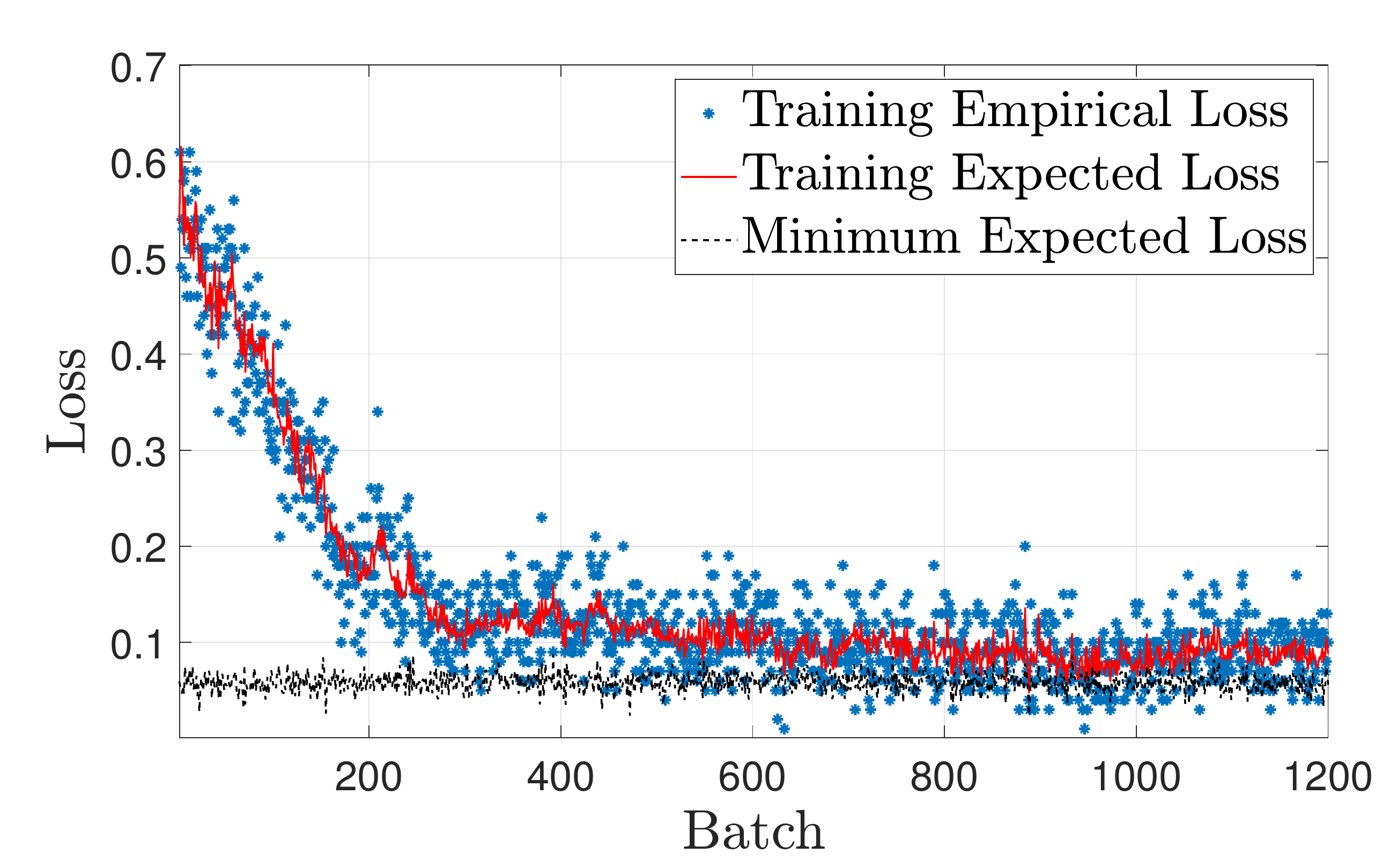}
\caption{The training loss of the QNN in Figure \ref{fig:QNN exp} using randomized QSGD. The loss is averaged over each batch of size $100$ samples. The figure shows the loss vs. batch number as the QNN is updated using Algorithm \ref{alg:QSGD}.  The minimum expected loss is calculated using Lemma \ref{lem:optimal loss}.}
\label{fig:Numerical Results}
\end{figure}


As seen in Figure \ref{fig:Numerical Results}, training loss converges to its minimum as more batches of samples are used, showing that randomized QSGD converges without the need for exact copies. 

To assess the impact of replicated data, we repeat the experiment, this time by allowing unlimited exact copies of samples. Therefore, instead of randomized QSGD, we train the QNN by exactly computing the gradient as in \eqref{eq:Delta Loss}. Figure \ref{fig:QNN Trace} shows the training loss for each batch of samples.
Comparing results in Figure \ref{fig:QNN Trace} (unlimited replication) with Figure \ref{fig:Numerical Results} (no replication), we observe a near-optimal training of the QNN using the randomized QSGD.   
\begin{figure}
    \centering
    \includegraphics[scale =0.3]{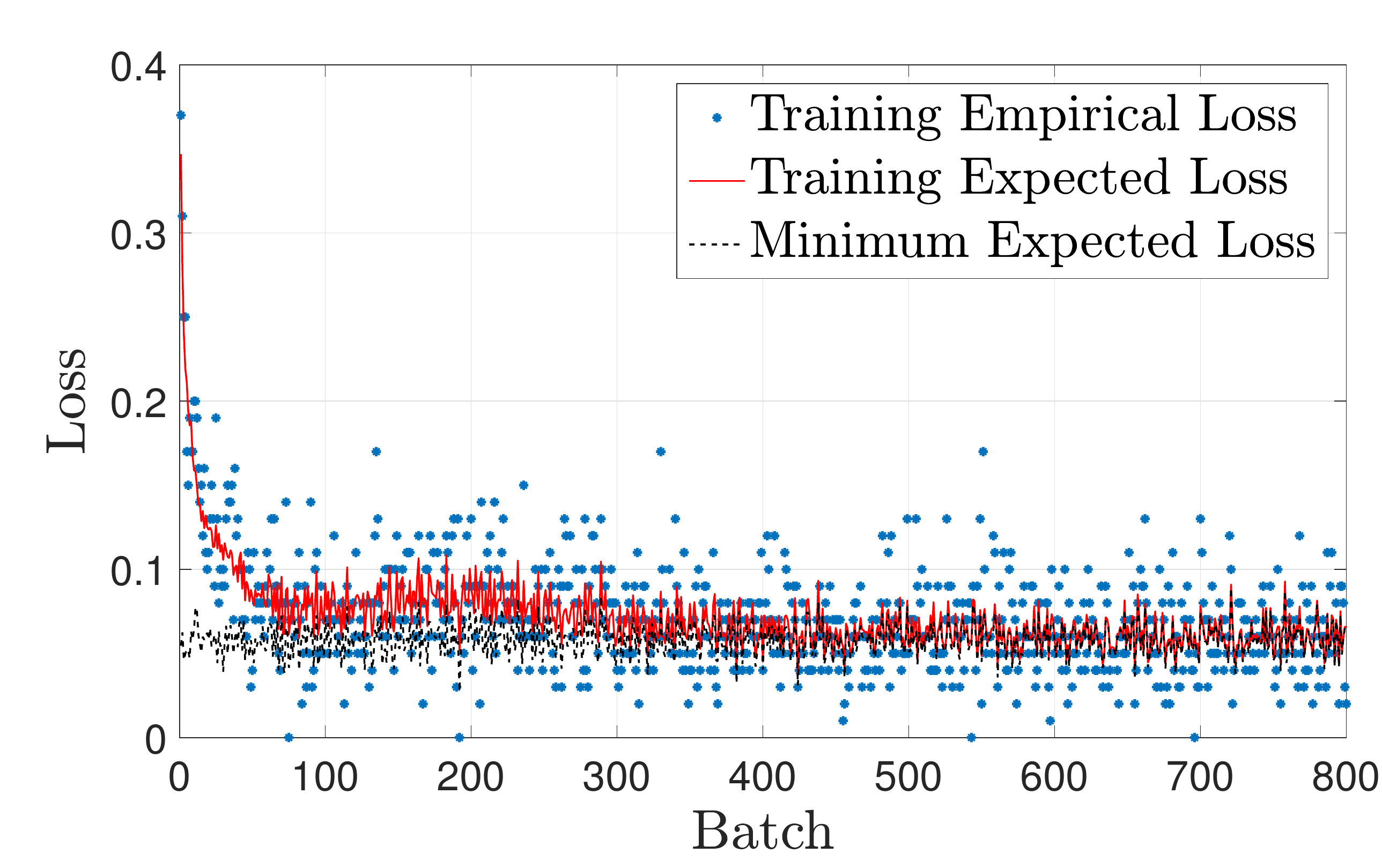}
    \caption{The training loss of the QNN in Figure \ref{fig:QNN exp} using exact computation of the gradient.}
    \label{fig:QNN Trace}
\end{figure}

\subsection{Comparison of Loss Function Values}

Once the QNN is trained, we generate a new set of $100$ samples to test the performance of the QNN. Table \ref{tab:Test Resuls} shows the resulting accuracy of the network, as compared to the optimal accuracy from Lemma \ref{lem:optimal loss} and the variant of the experiment with replicated samples. The resulting accuracy for our randomized QSGD is observed to be close to optimum. 
\begin{table}[ht]
    \centering
    \begin{tabular}{c|c c c}
    \toprule
                       &    QNN & QNN (gradient comp.)  & Optimal  \\
                       \hline 
     Acc.     &  $91\% \pm 1 $ &   $93.48\% \pm 1$  & $93.51 \%$\\
     \bottomrule
    \end{tabular}
    \caption{Validation of accuracy of the trained QNN in Figure \ref{fig:QNN exp} for classifying test samples. The first column is for the QNN trained with randomized QSGD. The second column is for the QNN trained using direct computation of the gradient. The third column is the optimal accuracy derived using Lemma \ref{lem:optimal loss}. }
    \label{tab:Test Resuls}
\end{table}


\section*{Concluding Remarks}\label{sec:conclusions}

In this work, we addressed two key shortcomings of conventional QNNs -- the need for replicated data samples for gradient computation and the exponential state space of QNNs. To address the first problem, we present a novel randomized SQGD algorithm, along with detailed theoretical proofs of correctness and performance. To address the exponential state space, we propose the use of band-limited QPs, showing that the resulting QNN can be trained in time linear in the number of QPs. We present detailed theoretical analyses, as well as experimental verification of our results. To this end, our work represents a major step towards realizable, scalable, and data-efficient QNNs.

\section*{Acknowledgements}
This work was supported in part by NSF Center on Science of Information Grants CCF-0939370 and NSF Grants CCF-2006440, CCF-2007238, OAC-1908691, and Google Research Award.

\bibliography{AAAI_main}
\end{document}